%% file: powerdomains.tex
\documentclass[copyright,creativecommons]{eptcs}
\providecommand{\event}{MFPS 2021} 
\usepackage{breakurl}             
\usepackage{underscore}           
\usepackage{graphicx}

\def\titlerunning{Two Guarded Recursive Powerdomains for Applicative Simulation}
\def\authorrunning{R.E. M{\o}gelberg \& A. Vezzosi}

\newcommand{\thanksAll}{\thanks{This work was supported by a research grant (13156) from VILLUM FONDEN.}}
\input{prelude.tex}

\begin{document}
\title{Two Guarded Recursive Powerdomains for Applicative Simulation}

\author{Rasmus Ejlers M{\o}gelberg\thanksAll
  \institute{Department of Computer Science\\ IT University of Copenhagen \\ Denmark} 
  \email{mogel@itu.dk}
\and
Andrea Vezzosi\footnotemark[1]
  \institute{Department of Computer Science\\ IT University of Copenhagen \\ Denmark} 
  \email{avez@itu.dk}
}

\maketitle

\begin{abstract} 
Clocked Cubical Type Theory is a new type theory combining the power of guarded
recursion with univalence and higher inductive types (HITs). This type theory
can be used as a metalanguage for \emph{synthetic guarded domain theory} in which one 
can solve guarded recursive type equations, also with negative variable occurrences,
and use these to construct models for reasoning about programming languages. 
Combining this with HITs allows for the use of type constructors familiar from set-theory 
based approaches to semantics, such as quotients and finite powersets in these models.

In this paper we show how to reason about the combination of finite non-determinism and 
recursion in this type theory. Unlike traditional domain theory which takes an ordering of programs as
primitive, synthetic guarded domain theory takes the notion of computation step as primitive
in the form of a modal operator. We use this extra intensional information to define 
two guarded recursive (finite) powerdomain constructions differing in the way 
non-determinism interacts with the computation steps. As an example application of these
we show how to prove applicative similarity a congruence 
in the cases of may- and must-convergence for the untyped lambda calculus with finite 
non-determinism. Such results are usually proved using operational
reasoning and Howe's method. Here we use an adaptation of a denotational method 
developed by Pitts in the context of domain theory.
\end{abstract}

\section{Introduction}
\label{sec:intro}


Over the past 20 years, step-indexing techniques~\cite{Appel:M01} have become one of the most used tools for
constructing operational models of programming languages with combinations
of advanced features such as recursive types, polymorphism, concurrency and non-determinism. 
Often such models are beyond the scope of traditional domain theoretic 
techniques, and also have the additional benefit of being 
more elementary. Guarded recursion is an abstract form of step-indexing, in which the explicit
steps are replaced by abstract computation steps in the form of a delay modality $\later$. This 
relieves the user of the book-keeping involved in explicit step-indexing and reveals the underlying 
structure that makes these models work in the form of an introduction $X \to \later X$, a guarded fixed point 
combinator of type $(\later X \to X) \to X$ and solutions to guarded recursive domain equations. In its multiclocked
version, where the delay modality $\later^\kappa$ is indexed by a clock $\kappa$ 
and clocks can be universally quantified, 
guarded recursion can moreover be used to encode coinductive types in type theory, allowing 
productivity requirements on these to be encoded in types~\cite{atkey13icfp}. 

Clocked Cubical Type Theory (CCTT)~\cite{kristensen2021model} is a type theory combining 
multiclocked guarded recursion 
with features from cubical type theory, in particular univalence and higher inductive types (HITs). The latter are 
a form of inductive types defined not only by constructors, but also by equations. HITs have been used to 
construct topological spaces such as the circle and the torus in type theory, but can also be used for free structures, 
such as the free group on a set. In computer science, free structures can be used to form the monads
generated by algebraic theories. For example, the finite powerset monad, often used to model finite non-determinism,
can be generated by a binary union operation plus axioms of associativity, commutativity, and idempotency, and can
therefore be naturally represented as a HIT~\cite{BasoldGW17,FGGW18}. 

Combining HITs with guarded recursion provides a powerful metatheory in which one can reason about 
programming languages and programs. This paper presents a worked example of this. We study the untyped 
lambda calculus with finite non-determinism, and show how to construct a model of this in CCTT using a guarded
recursive type. The model construction takes as parameter a monad $\gMon$ with a 
union operation $\gCup : \gMon X \times \gMon X \to \gMon X$ for modelling non-determinism, 
as well as a step operation $\gstep : \later^\kappa (\gMon X) \to \gMon X$, which in combination with the fixed point
operator allows us to model recursion. 

We present two instantiations for $\gMon$, corresponding to two different notions of observation 
on non-deterministic programs. The first describes a notion of non-deterministic computation where all
possible branches of a computation are executed in parallel and we can observe all possible values that 
occur along the way, even if there are diverging branches of the computation. This monad corresponds to 
may-convergence and can be characterised as being generated by the operations $\gCup$  and $\gstep$
with no equations between them. 

The second instantiation corresponds to a notion of computation where all branches are evaluated in parallel, but 
partial results are only available when all branches have terminated. This is simply the composition $\Lift \Pfin$ 
of the free monad $\Lift$ generated by the step operation and the finite powerset monad $\Pfin$. It turns out that 
this composition is not itself a monad, but does have a sequencing operation sufficient for the purposes of this 
paper. We claim that this composition corresponds to must-convergence.

As an example application of this model, and to substantiate the claim that these constructors correspond 
to may- and must-convergence, respectively, we apply the model to a classical problem in lambda calculus,
namely that of proving that applicative similarity is a congruence. This was first proved by 
Abramsky~\cite{abramsky1990lazy} for the 
lazy lambda calculus (without non-determinism) using domain theory and Stone duality, and this method
has since been extended to calculi with non-determinism~\cite{ong1993non}. Here we use a different method
due to Pitts~\cite{pitts1997note} who used a domain theoretic model and a relation between syntax and 
semantics. We extend this proof to non-determinism for both may- and must-equivalence, and adapt it 
to guarded recursion. 

The two instantiations of $\gMon$ mentioned above are what we consider as examples of guarded recursive 
powerdomains. In classical domain theory~\cite{domain:handbook} a powerdomain is a domain theoretic 
correspondent to the powerset construction in set theory. A number of different powerdomains exist, each 
characterising a different notion of observation. One way of characterising the difference between these is 
in terms of enriched algebraic theories which allow them to be classified in terms of inequations such as 
$x \leq x \cup y$. In the case of guarded recursion the difference between the two guarded powerdomains
studied here can be expressed in terms of equalities describing the interaction between $\gCup$ and $\gstep$. 
This paper can therefore also be read as a first study of the interaction of algebraic effects and guarded recursion.


\subsection{Synthetic guarded domain theory}

Until now most applications of guarded recursion and step-indexing have used these for 
operational reasoning for programming languages. 
While these techniques are very useful for proving properties of programs, we believe that 
there is a need for also developing guarded recursion as a tool for constructing denotational semantics. 
Denotational methods have the benefit of often being more modular than the operational methods,
and often reveal the foundational mathematical building blocks of programming languages. Denotational
semantics often inspire new programming constructions or languages, as exemplified in 
monads~\cite{moggi1991notions}, 
runners for computational effects~\cite{uustalu2015stateful,ahman2020runners}, 
homotopy type theory~\cite{hottbook}, or even guarded recursion itself~\cite{ToT}. 

Using guarded recursion for denotational semantics has several possible benefits over domain theory.
The first is that it appears to be more expressive than domain theory as illustrated by the many uses 
of step-indexing for advanced programming languages. Another is that step-indexing and guarded recursion
by many are considered more elementary tools. A third is that guarded recursion appears to be more 
amenable to effective formalisation in type theory and proof assistants, although some formalisations 
of classical domain theory do exist~\cite{benton2009some,dockins2014formalized}, 
and in particular recent progress on such a formalisation in
HoTT seems promising~\cite{dejong:CSL:2021}. 

Initial steps towards such a \emph{synthetic guarded domain theory} were taken by Birkedal, M{\o}gelberg 
and Paviotti~\cite{GDTT:FPC,paviottiPCF} who showed how to construct models of the programming 
languages PCF and FPC modelling recursion in these as guarded recursion, and proving the adequacy of 
these models entirely in a type theory with guarded recursion. This paper can be viewed as an extension of 
these works to non-determinism. Perhaps the main disadvantage of guarded recursion compared to domain
theory is the intensional nature, allowing the model to distinguish between computations that produce the 
same result in a different number of steps. The present paper shows that using universal quantification over 
clocks allows to localise the steps and to prove properties that do not refer to steps, using the models.


\subsection{Related work}

Most proofs of applicative similarity being a congruence use operational 
arguments~\cite{lassen1998similarity,lassen1998relational}, in particular Howe's 
method~\cite{howe1989equality}. More recently, an abstract version of Howe's method
has been developed~\cite{dal2017effectful} to handle languages with algebraic effects in 
a uniform way. This method uses domain theory to handle recursion. It would be interesting to 
see if the method described here generalises to a similar uniform method for computational
effects, but this requires first developing a theory of algebraic effects in guarded type theory. 

Step-indexing and guarded recursion based operational techniques have previously been used for languages 
with non-determinism. For example, Schwinghammer et al.~\cite{schwinghammer2013step}
construct an operational model for reasoning about a typed programming language with recursive types, 
polymorphism and non-determinism and use it to prove contextual equivalences of programs. Bizjak
et al.~\cite{Bizjak:TLCA:2014} show how to construct a similar model using guarded 
recursion and topos logic. These works use complex operational techniques including $\top\top$-closure.
Our goal is different, namely to develop a theory of denotational semantics in a type theory with guarded 
recursion.

The above mentioned works on 
non-determinism~\cite{lassen1998similarity,lassen1998relational,schwinghammer2013step,Bizjak:TLCA:2014} 
study countable non-determinism, rather than finite non-determinism. 
This is generally considered a harder problem. For example, this forces the step-indexing used by 
Schwinghammer et al.~\cite{schwinghammer2013step} to 
be transfinite, whereas the underlying model of the Clocked Cubical Type Theory is based on natural number 
step-indexing. Likewise, defining powerdomains in domain theory for countable non-determinism is much harder 
than the finite case~\cite{apt1986countable,di1995uncountable}. We discuss the possibility of extending our 
approach to countable non-determinism in Section~\ref{sec:conclusion}. 

As described in Section~\ref{sec:cctt}, our partiality monad $\Lift$ is strongly related to the 
coinductive partiality monad, and our use of it is similar to previous uses in semantics of 
recursion~\cite{benton2009some,capretta:lifting,danielsson2012operational}. Interaction trees~\cite{xia2019interaction}
are a general data structure combining the coinductive partiality monad with computational effects. Our guarded 
powerdomain monads can perhaps be seen as a form of guarded interaction trees for non-determinism, except 
that the use of HITs allows us to consider these up to an equational theory. 


\subsection{Overview}

The paper is organised as follows: We first recall the basics of Cubical Type Theory in Section~\ref{sec:ctt}, 
in particular path types and higher 
inductive types. Section~\ref{sec:cctt} then recalls Clocked Cubical Type Theory, 
the extension of Cubical Type Theory 
with multiclocked guarded recursion. Section~\ref{sec:may:pow:dom} defines the guarded powerdomain 
for may-convergence and Section~\ref{sec:app:may:bisim} presents our proof that applicative may-similarity
is a congruence using a denotational model. Section~\ref{sec:must:pow:dom} defines a  
guarded powerdomain for must-convergence, and Section~\ref{sec:app:must:bisim} 
presents our proof that applicative must-similarity is a congruence. We conclude 
in Section~\ref{sec:conclusion}.

%
%


\section{Cubical type theory}
\label{sec:ctt}

Cubical Type Theory (CTT)~\cite{CTT} is a variant of Homotopy Type Theory 
(HoTT)~\cite{hottbook} based on the 
cubical model of the univalence axiom, and specifically designed to compute 
with univalence. It moreover has the benefit of combining more easily with 
guarded recursion than HoTT, which was the reason for using it as a base for 
Clocked Cubical Type Theory as we shall describe in Section~\ref{sec:cctt}. 
Reading this paper does not require deep knowledge of CTT, and this section
recalls the basic notions from CTT and HoTT that we shall need. 

Perhaps the most fundamental difference between CTT and Martin-L{\"o}f type theory
is that the identity type in the latter is replaced in CTT by a type of paths 
$\mathsf{Path}_A(x,y)$ 
between two elements $x, y : A$.
We will often write the path type infix as $x = y$, and say that 
$x$ and $y$ are path equal if there is an element of $x = y$.  Cubical Type Theory
represent paths as maps from an abstract interval type $\I$, with endpoints $0$ and $1$.
In particular a lambda abstraction like $\lambda i. t$ will build a
path of type $t\subst 0i = t\subst 1i$.
This allows more canonical proofs of equality than just
reflexivity, and so to give computational content to principles like
function extensionality and univalence.
Path equality is still substitutive, in the sense that any element of
type $P(x)$ can be transported along a path $x = y$ to construct an
element of $P(y)$. In the following we will not rely on details about the specific
primitives of CTT, which can be found in \cite{CTT}, and
\cite{CubicalAgda} for their incarnation in the Agda proof assistant.

Types can be classified according to the complexity of their
path equality: We say a type is a \emph{mere
  proposition} if any two elements are path equal, and that a type is
a \emph{homotopy set}, or simply a \emph{set}, if its path equality
type is a mere proposition. These predicates can be expressed 
in the type theory as types $\isProp(A)$ and $\isSet(A)$  for a type $A$.
Given any universe of types $\Univ$, we can form a universe of mere
propositions $\Prop$ whose elements are pairs of an element of
$\Univ$ and a proof that it is propositional. For $A : \Prop$ we will
often write $A$ itself rather than its first projection.

Cubical Type Theory also supports Higher Inductive Types (HITs), which allow
to define an inductive type by declaring constructors also for its
path equality, rather than only for its elements.
For example, the propositional truncation $\Trunc A$ of a type $A$ is
defined as a higher inductive type with the following constructors
\begin{align*}
|-| & : A \to \Trunc A & 
\mathsf{squash} & : \Pi (x\,y : \Trunc A).\,x = y
\end{align*}
This defines the least proposition extending $A$
in the sense that any map $f : A \to B$ into a proposition $B$, 
defines a unique map $\overline f : \Trunc{A} \to B$, such that 
$\overline f(|a|) = f(a)$ for all $a$. If $B$ is a set, then such an extension
still exists if $\Pi (x\,y : A). f\,x = f\,y$~\cite{kraus14}.
For any two propositions $A$ and $B$ their conjunction $A \wedge B$ is
given by the cartesian product $A \times B$, their disjunction $A \vee
B$ by truncating their disjoint union $\Trunc{A + B}$, while the true
and false propositions are given by the unit and the empty type.
Univalence implies that any two logically equivalent propositions are
equal, so associativity and commutativity of disjunction and other
such laws hold as path equalities.
Given a predicate $P : A \to \Prop$, universal quantification $\forall
a : A. P\,a$ is given by the dependent function type $\Pi
(a:A). P\,a$, while existential quantification $\exists a : A. P\,a$
is given by truncating the dependent pair type $\Trunc{\Sigma a :
  A. P\,a}$ as generally it will not be propositional otherwise. As
$\Trunc{Q}$ and $Q$ coincide when $Q$ is a proposition, so do $\exists
a : A. P\,a$ and $\Sigma a : A. P\,a$ when $P\,a$ uniquely determines
$a$. 

\subsection{Finite Powerset}
\label{sec:pfin}

The finite powerset $\Pfin(A)$~\cite{BasoldGW17,FGGW18} of a type is another example of a higher 
inductive type defined by the following constructors 
\begin{align*}
 \{-\} & : A \to \Pfin(A)\\
 \cup & : \Pfin(A) \to \Pfin(A) \to \Pfin(A)\\
 \mathsf{assoc} & : \Pi (X\,Y\,Z : \Pfin(A)).\,X \cup (Y \cup Z) = (X \cup Y) \cup Z\\
 \mathsf{comm} & : \Pi (X\,Y : \Pfin(A)).\,X \cup Y = Y \cup X \\
 \mathsf{idem} & : \Pi (X : \Pfin(A)).\,X \cup X = X 
\end{align*}
plus two equalities ensuring that $\Pfin(A)$ is a set \cite{kristensen2021model}. Note that we restrict ourselves to 
\emph{non-empty} finite powersets.
We say that a pair of a type $B$ and a binary operation $f : B \to B
\to B$ is a \emph{join-semilattice}\footnote{This is a slight misuse of terminology, since join-semilattices are 
usually assumed to also have a unit} 
if $B$ is a set and $f$ is associative, commutative
and idempotent.
It can then be shown that $\Pfin(A)$ is the free join-semilattice generated by $A$, and as such maps
$\Pfin(A) \to B$ which preserve $\cup$ correspond to maps $A \to B$ for any join-semilattice $(B,f)$.
We use this to define the membership predicate $x \in
X$, as $\Prop$ forms a join-semilattice with disjunction. Membership satisfies the
following equations
\begin{align*}
  x \in \{ y \} & = \Trunc{x = y}  & 
  x \in (X \cup Y) & = x \in X \vee x \in Y
\end{align*}
If $X : \Pfin(A)$ we write $\forall a \in X.\,Q(a)$ to mean $\Pi (a :
A).\,a \in X \to Q(a)$ and similarly for $\exists a \in X. Q(a)$.

The finite powerset also supports the structure of a monad, and in
particular, given $f : A \to \Pfin(B)$ we write $\cup_{a \in X} f\,a$
for the bind operation, defined using the free 
join-semilattice structure.  Given $f : A \to B$ we write $\Pfin(f) : \Pfin(A) \to
\Pfin(B)$ for the functor action of the finite powerset, defined as
$\Pfin(f)\,X \defeq \cup_{a \in X} \{f\,a\}$.

Finally, recall~\cite[Lemma~4.1]{mogelbergPOPL2019} that if $f : A \to
B$, $X : \Pfin(A)$, and $b : B$, then 
\begin{equation}
\label{eq:image:powset}
b \in \Pfin(f)\,X \simeq \exists a \in X. f\,a = b
\end{equation}

\section{Clocked cubical type theory}
\label{sec:cctt}
\begin{figure}
\begin{mathpar}
  \inferrule*
  {\kappa : \clocktype \in \Gamma}
  {\wfcxt{\Gamma, \tickA : \kappa}{}}
  \and
  \inferrule*
  {\hastype{\Gamma, \timeless{\Gamma'}}{t}{\latbind{\tickA}{\kappa} A}\\ \wfcxt{\Gamma,\tickB:\kappa,\Gamma'}}
  {\hastype{\Gamma,\tickB: \kappa,\Gamma'}{\tapp[\tickB] t}{A\toksubst{\tickB}{\tickA}}}
  \and
  \inferrule*
  {\hastype{\Gamma,\tickA:\kappa}{t}{A}}
  {\hastype{\Gamma}{\tabs{\tickA}{\kappa} t}{\latbind{\tickA}{\kappa} A}}
  \and
    \inferrule*
  {\hastype{\Gamma,\kappa : \clocktype}{t}{A}}
  {\hastype{\Gamma}{\Lambda\kappa. t}{\forall \kappa . A}}
  \and
  \inferrule*
  {\hastype{\Gamma}{t}{\forall \kappa . A}\\
    \hastype\Gamma{\kappa'}\clocktype}
  {\hastype{\Gamma}{t [\kappa']}{A \subst{\kappa'}{\kappa}}}
\end{mathpar}
\caption{Selected typing rules for Clocked Cubical Type Theory \cite{kristensen2021model}.
  The telescope $\timeless{\Gamma'}$ is
  composed of the timeless assumptions in $\Gamma'$, i.e. interval variables
  and faces (as in Cubical Type Theory) as well as clock variables.}
\label{fig:later:typing}
\end{figure}

Clocked Cubical Type Theory \cite{kristensen2021model} extends Cubical Type Theory
with the constructions of Clocked Type
Theory~\cite{bahr2017clocks,clottmodel}, a type theory with Nakano
style guarded recursion, multiple clocks and ticks. This section recalls each of these
concepts, but in a simplified form, omitting constructions related to 
tick irrelevance.

The fundamental notion in guarded recursion is that of a time step on a clock. 
Clocks are introduced as assumptions of the form $\kappa : \clocktype$ in the context,
and time steps are represented as tick assumptions of the form $\tickA : \kappa$.
The type $\latbind\tickA\kappa A$ classifies computations that in the next time step
(as represented by the tick $\tickA$) return elements of $A$. When $\tickA$ does 
not appear in $A$ we simply write $\later^\kappa A$ for this type.
Elements of $\latbind\tickA\kappa A$ are introduced by \emph{tick}
abstraction $\tabs \tickA \kappa t$ and eliminated by tick application
$\tapp t $. The rules are similar to those for function types, except that
tick application requires the eliminated term $t$ not to depend on $\tickA$, nor any
of the variables bound before $\tickA$. This rules out terms like $\lambda
x. \tabs\tickA\kappa{\tapp {\tapp x}} : \later^\kappa
\later^\kappa A \to \later^\kappa A$, which collapse two steps
into one. Interval variables are considered timeless 
and therefore exempt from this restriction, which is necessary to
prove that tick application preserves equalities:
\begin{align} \label{eq:delay:extensionality}
(\lambda p.\, \tabs\tickA\kappa{\lambda i.\, (\tapp {p\,i})}) : x =_{\latbind{\tickA\,}{\,\kappa} A} y  \to
(\laterAbs\tickA\kappa{ \tapp x =_A \tapp y}) 
\end{align}
In fact, the above map is an equivalence of types, and this extensionality 
principle is one of the main reasons CTT is used rather than
HoTT. One consequence is that $\later$ preserves 
truncation levels. In particular, if $\latbind\tickA\kappa{\isProp(\tapp A)}$ then also 
$\isProp(\latbind\tickA\kappa(\tapp A))$~\cite[Lemma~3.1]{mogelbergPOPL2019} 
and similarly for sets. 

The delay type allows to safely introduce a fixpoint combinator 
$\fix^\kappa$ of type $(\later^\kappa A \to A) \to A$ and
satisfying the path equality $\fix^\kappa\,t = t\,(\tabs\tickA\kappa
{\fix^\kappa t})$ for any $t$.
We can use $\fix^\kappa$ to define \emph{guarded recursive types}, i.e. ones 
where the recursive occurrences are guarded by $\later^\kappa$. 
An example is the partiality monad mapping $A : \Univ$ to 
$\L^\kappa\, A \defeq \fix^\kappa (\lambda (X :
\later^\kappa \Univ).\, A + \laterAbs\tickA\kappa {\tapp X})$.
The path equality between this type and its unfolding gives 
rise to a type equivalence
\[
 \L^\kappa A \equi A + \later^\kappa\L^\kappa A
\]
We use $\liftnow{}$ and $\liftstep{}$ to denote the two inclusions into $\L^\kappa\, A$ up to the
above equivalence.
\begin{align*}
\liftnow{} & : A \to \L^\kappa A & 
\liftstep{} & : \later^\kappa {\L^\kappa A} \to \L^\kappa A
\end{align*}
An element of $\L^\kappa A$ represents a possibly non-terminating computation of an 
element of $A$. For example, the element 
$\perp = \fix^\kappa (\liftstep{})$ represents divergence.
We say that $(A,\delta)$ is a \emph{delay algebra} if $\delta$ has
type $\later^\kappa A \to A$. The pair $(\L^\kappa
A,\liftstep{})$ is the free delay algebra generated by $A$ in the sense that any 
map $f : A \to B$ where $(B,\delta)$ is a delay algebra defines a unique 
$\overline{f} : \L^\kappa A \to B$ such that
\begin{align*}
  \overline{f}(\liftnow{a}) &= f\,a & \overline{f}(\liftstep{x}) &= \delta (\tabs\tickA\kappa {\overline{f}(\tapp x)})
\end{align*}
The monad structure is then defined with $\liftnow{}$ as the unit, and multiplication 
$\mu_\L : \L^\kappa \L^\kappa A \to \L^\kappa A$ defined as the unique (delay algebra)-homomorphism 
extending the identity.


The \emph{clock quantification} type former $\forall \kappa. A$ 
is introduced by clock abstraction, and eliminated by application to a clock. 
It behaves much like a $\Pi$-type, except that $\clocktype$ is not a type, but has 
a similar status to the interval $\I$. 
Clock quantification localises guarded recursion on a clock, and in particular 
supports a map $\force : \forall \kappa. \later^\kappa A \to \forall
\kappa. A$, inverse to $\lambda x.\,\lambda \kappa. \tabs
\tickA \kappa {\capp x}$, allowing to safely eliminate $\later^\kappa$.

The main use case for clock quantification is to encode coinductive
types. For example $\forall \kappa. \L^\kappa A$ is the
final coalgebra for the functor $F(X) = A + X$, if $A$ is \emph{clock irrelevant},
i.e., if the canonical map $A \to \forall \kappa. A$ is an equivalence. 
The notion of clock irrelevance is closed under all basic type formers as
well as inductive types and (under certain restrictions~\cite{kristensen2021model}) 
higher inductive types. In particular, the inductive types used in this paper to 
represent syntax are all clock irrelevant. This encoding of coinductive
types is originally due to Atkey and McBride~\cite{atkey13icfp} and 
presumes the existence of a clock constant $\kappa_0$, which we achieve by 
just initially assuming $\kappa_0 : \clocktype$.

More generally for every (indexed)
functor $F$ which commutes with clock quantification we have that
$\forall \kappa. \fix^\kappa (\lambda X.\,
F(\laterAbs\tickA\kappa{\tapp X}))$ is the final coalgebra of $F$,
i.e. its coinductive fixpoint \cite{kristensen2021model}. 
The collection of functors commuting with clock quantification is closed under a
long list of constructors including truncations, finite powersets and sum types
as expressed in the following type equivalences.
\begin{align*}
  \forall \kappa. \Trunc A & \simeq \Trunc {\forall \kappa. A}
  & \forall \kappa. \Pfin(A) & \simeq \Pfin(\forall \kappa. A)
  & \forall \kappa. (A + B) & \simeq \forall \kappa. A + \forall \kappa. B
\end{align*}

\section{A powerdomain for may-convergence}
\label{sec:may:pow:dom}

Define the may powerdomain as the unique solution to the guarded recursive equation
\[ \Pmay (A) \equi \Pfin(A + \later^\kappa \Pmay(A)) \]
Formally, $\Pmay$ can be defined as a fixed point $\fix^\kappa (\lambda (X :
\later^\kappa \Univ).\, \Pfin(A + \laterAbs\tickA\kappa {\tapp X}))$ similarly to the definition
of $\L^\kappa A$, but in the rest of this paper we will not give such definitions explicitly. 
The type constructor $\Pmay$ comes equipped with the following operations
\begin{align*}
 \cup & : \Pmay(A) \to \Pmay(A) \to \Pmay(A) &
 \maynow & : A \to \Pmay(A)  &
 \maystep & : \later^\kappa \Pmay(A) \to \Pmay(A)
\end{align*}
where $\cup$ is inherited from $\Pfin$ and therefore defines a join-semilattice and 
\begin{align*}
 \maynow(a) & = \{ \inl (a)\} & 
 \maystep(a) & = \{ \inr (a)\}
\end{align*}
defines a unit and a delay algebra structure. In particular, this means that 
$\Pmay(A)$ can represent diverging computations ($\bot = \fix^\kappa(\maystep\!)$), 
as well as values. An element of $\Pmay(A)$ may also have both converging and 
diverging branches, as for example $\{a, \bot\}$.

The next lemma states that for a set $A$,  $\Pmay(A)$ is the free algebra for the theory
combining delay and union with no interaction between the two.
\begin{lemma} \label{lem:pmay:universal:prop}
 Let $A$ be a set, let $B$ be a set with both a join-semilattice structure and a delay algebra structure,
 and let $f : A \to B$. Then there is a unique map $\Pmay(A) \to B$ extending $f$ 
 and commuting with the join-semilattice and delay-algebra structures. 
\end{lemma}
In terms of algebraic theories, $\Pmay$ can therefore be seen as generated by the theories 
of join-semilattices and delay-algebras with no operations between them. As a special case of 
the lemma one can define a bind operation mapping $a : \Pmay(A)$ and 
$f : A \to \Pmay B$ to $\bindnolambda af: \Pmay B$ which then equips $\Pmay$ with a monad 
structure with unit given by $\maynow$. Bind moreover commutes with these operations as in the 
following equations
\begin{align}
(\bindnolambda {\maystep a}f) 
& = \maystep {\tabs\tickA\kappa{(\bindnolambda{\tapp a}f)}} \label{eq:bind:comm:step} \\
(\bindnolambda {(a \cup b)}f) 
& = (\bindnolambda {a}f) \cup (\bindnolambda {b}f) \label{eq:bind:comm:cup}
\end{align}

\section{Applicative may-simulation}
\label{sec:app:may:bisim}

The type constructor $\Pmay$ should be seen as a guarded recursive 
powerdomain for may-convergence. Intuitively, this is true because an element
of $\Pmay(A)$ describes a set of values (of type $A$) that the computation
has returned now, and a set of computations that we can choose to further evaluate,
but may also choose not to, if we are only interested in testing if a program 
may evaluate to a particular value. This should be seen in contrast to the 
must-powerdomain of Section~\ref{sec:must:pow:dom} which will force all branches of 
a computation tree to be evaluated fully before the result values can be inspected. 
In this section we substantiate that claim 
by using $\Pmay$ to prove applicative may-similarity a congruence. 

We start by recalling the untyped lambda calculus with binary
non-determinism and the notion of applicative 
may-similarity. We use informal binding notation here for 
readability, using the grammar
\begin{align*}
 M, N & :: = M \, N \mid \olam x M \mid M \oor N
\end{align*}
for terms. This could for example be implemented more formally as an inductive
family using de Bruijn indices. 
Note the use of $\olambda$ to distinguish it from the meta-level lambda. 
A value is a closed term of the form $\olam x M$, and we shall use 
$\Lambda$ and $\Val$ for the types of closed lambda terms and values 
respectively, which we will assume are sets (and indeed are if formalised using 
de Bruijn indices). These are moreover clock irrelevant, which can be shown by 
embedding them into the inductive types of all (also open) terms, and 
the fact that all inductive types are clock irrelevant~\cite{kristensen2021model}.

We define two operational semantics. The first one is 
a big-step operational semantics formulated as a relation
$\BSmay : \Lambda \times \Val \to \Prop$ defined inductively in the standard way
\begin{mathpar}
  \infer{V = W}{V \BSmay W} \and
  \infer{M \BSmay \olam x M' \\ N \BSmay V' \\ M' \subst{V'}x \BSmay V}{M N \BSmay V} \and
  \infer{M \BSmay V \hvee N \BSmay V}{M \oor N \BSmay V} 
\end{mathpar}
%
%
One can also similarly define
a small-step operational semantics and prove this equivalent to the 
big-step semantics using standard methods. 

The second operational semantics is less familiar but has the benefit of 
being suitable for guarded recursive reasoning. It uses the monad
$\Pmay$ to model recursion and non-determinism, but since parts of the 
development in this section will be reused later, we will define 
the operational semantics relative to a monad $\gMon$, 
which can be instantiated to be $\Pmay$. We will assume that 
$\gMon$ has operations $\cup$ and $\gstep$ providing $\gMon$ with
a join-semilattice structure as well as a delay-algebra structure satisfying the
equations (\ref{eq:bind:comm:step}) and (\ref{eq:bind:comm:cup}). 
In fact we shall see later that not all axioms for monads are needed 
for our developments. 

Using this, we define an evaluation function $\eval : \Lambda \to \gMon (\Val)$ as 
\[
{\arraycolsep=1.2pt
\begin{array}{l l c l}
  \eval & (\olam x M) &=& \gnow\,(\olam x M)\\
  \eval & (M\, N) &=& \eval\,M \sequencing \lambda (\olam x M').\,
  \eval\,N \sequencing \lambda V.\,
  \gstep (\tabs\tickA\kappa{\eval \, (M' [V/x])}) \\
  \eval & {(M \oor N)} & =& \eval\,M \cup \eval\,N
\end{array}
}
\]
where $\gnow$ refers to the unit of the monad $\gMon$. Note that the match on $(\olam x M')$
in the application case is exhaustive, since values can only be lambda abstractions.
To lift this to a big step relation, we will assume given a lifting of $\gMon$
to predicates as follows.

\begin{definition} \label{def:pred:lifting}
 A lifting of a monad $\gMon$ to predicates is a function that 
 maps a predicate $P : A \to \Prop$ to a predicate $\gLift(P) : \gMon (A) \to \Prop$
 satisfying the following properties
 \begin{enumerate}
   \item $\gLift(P)(\gMon (f)(a)) = \gLift(P \circ f)(a)$ \label{item:Lift:comp}
   \item $\gLift(P)(\gnow(a)) = P(a)$  \label{item:Lift:now}
   \item $\gLift(P)(\gstep(a)) = \latbind\tickA\kappa{\gLift(P)(\tapp a)}$ \label{item:Lift:step}
   \item $\gLift(P)(\bind mxt) = \gLift(\lambda x . \gLift(P)(t))(m)$ \label{item:Lift:bind}
   \item $\gLift(P)(a \cup b) = \gLift(P)(a) \hmeet \gLift(P)(b)$ \label{item:Lift:cup}
 \end{enumerate}
\end{definition}

Note that a unique such lifting can be defined for $\Pmay$ using 
Lemma~\ref{lem:pmay:universal:prop} since $\Prop$ is a set,
and $\hmeet$ and $\later$ define a join-semilattice and delay-algebra
structure on $\Prop$, so $\gLift(P)$ can be defined
as the extension of $P$.

Specialising to $T = \Pmay$ we make the following definition.
\begin{definition}
 Let $\gLift$ be the lifting of $\Pmay$ to predicates, let $Q : \Val \to \Prop$ be a predicate on
 values and let $M : \Lambda$. Define 
 \[
   M \BSopmay[\kappa] Q  \defeq \gLift Q (\eval (M)) :  \Prop
 \]
\end{definition}
One can then define a more standard big step operational semantics as 
\begin{align*}
 M \BSopmay[\kappa] V & \defeq M \BSopmay[\kappa] (\lambda W . (V = W)) 
\end{align*}
Intuitively $M \BSopmay[\kappa] Q$
means that if $M$ terminates to a value, that value will satisfy $Q$. 
In particular, $M \BSopmay[\kappa] V$ will hold for any $V$ if $M$ diverges, 
unlike the statement $M \BSmay V$ which guarantees termination. 

To express the precise relationship between the two semantics, we
introduce a predicate of may-convergence on the powerdomain
$\Pmay$. Since termination cannot be expressed as a predicate on guarded recursive 
types directly, it must be expressed as a predicate on $\forall\kappa. \Pmay(-)$,
which captures global behaviour of $\Pmay$~\cite{GDTT:FPC}. If
$m : \forall \kappa. \Pmay(A)$ and $a : A$ define $m \Convmay a$ 
as the proposition inductively generated by the following introductions
\begin{mathpar}
  \inferrule*{ }{ \lambda \kappa. \maynow{(a)} \Convmay a}
  \and
  \inferrule*{ m \Convmay a }{\lambda \kappa. \maystep{(\tabs\alpha\kappa {\capp m})} \Convmay a}
  \and
  \inferrule*{ m \Convmay a \hvee m' \Convmay a }{\lambda \kappa. \capp m \cup \capp {m'} \Convmay a}
\end{mathpar}
This means that these rules should be read as constructors of a HIT which also has propositional 
truncation operators as a constructor. The relationship between the two semantics can then be
expressed as follows.

\begin{proposition} \label{prop:bs:may:equiv}
  The statements $M \BSmay V$ and $(\lambda \kappa. \eval (M)) \Convmay V$ are logically equivalent. 
\end{proposition}

Finally, the next lemma makes the intuition for $M \BSopmay[\kappa] Q$ stated above precise.
As the notation used in the lemma suggests, $\kappa$ can appear free in $Q^\kappa$.

\begin{lemma}\label{lem:convmay:elim}
  Let $A$ be clock irrelevant, $Q^\kappa$ a family over $A$ for each $\kappa$, and $m : \forall \kappa.\, \Pmay(A)$.
  The statements $\forall\kappa . \PmayPred Q^\kappa {(\capp m)}$ and $\forall a.\, m \Convmay a \to \forall\kappa . Q^\kappa(a)$ are logically equivalent. 
 As a consequence ${\forall\kappa . M\BSopmay[\kappa] Q^\kappa}$  is equivalent to 
 $M \BSmay V \to \forall\kappa . Q^\kappa(V)$.
\end{lemma}

%
%
%
%

\subsection{Applicative may-similarity}
\label{sec:appl:bisim}

We now recall the notion of applicative similarity, as originally studied by Abramsky~\cite{abramsky1990lazy} for the 
pure lambda calculus and adapt it to finite non-determinism in the case of may-convergence. 
Say that a relation $R$ on closed terms is an applicative may-simulation if $M R N$ and $M \BSmay \olam x {M'}$ implies
\[\exists N'.\,N \BSmay \olam y {N'} \hmeet
 (\forall (V : \Val) . M'\subst {V} x \,R\, N'\subst {V} x)
 \]
Applicative may-similarity is the greatest applicative may-simulation. 
We define this by universally quantifying a clock in a guarded recursive 
definition. First define 
\begin{align*}
 \appsimmayV[\kappa] & : \Val \to\Val \to \Prop \\
 \appsimmay[\kappa] & : \Lambda \to \Lambda \to \Prop \\
 \olam x M \appsimmayV[\kappa] \olam yN & \defeq \latbind\tickA\kappa {(\forall V \! :\! \Val . M\subst Vx \appsimmay[\kappa] N\subst Vy)} \\
 M \appsimmay[\kappa] N & \defeq M \BSopmay[\kappa] \lambda V . (\exists W. N \BSmay W \hmeet V \appsimmayV[\kappa] W)
\end{align*}
The statement $M \appsimmay[\kappa] N$ should be read as stating that if $M$ terminates, then also $N$ 
terminates, and 
moreover, applying the resulting terms to the same value results in the related results later. Note the 
asymmetry in the use of operational semantics: The evaluation of $M$ uses the guarded operational semantics
which ensures that if $M$ diverges, then $M \appsimmay[\kappa] N$ is true. On the other hand, once $M$ converges
to a value, $N$ must also converge, and expressing this requires the inductive operational semantics. The delay 
in the definition of $\appsimmayV[\kappa]$ ensures well-definedness: unfolding the definition of $\appsimmay[\kappa]$
in $\appsimmayV[\kappa]$ gives a guarded recursive definition.

Applicative similarity is extended to open terms by defining 
$M \appsimmay[\kappa] N$ to mean $M\sigma \appsimmay[\kappa] N\sigma$
for all substitutions $\sigma$ mapping all free variables in $M$ and $N$ to values. Finally, we localise the 
steps in the definition of $M \appsimmay[\kappa] N$ by universally quantifying $\kappa$ and thereby pass to
a coinductive type:
\[ M \appsimmay N \defeq \forall\kappa . M \appsimmay[\kappa] N
\]

\begin{lemma} \label{lem:appbisim:great}
 $\appsimmay$ is the greatest applicative may-simulation. 
\end{lemma}

We now proceed to prove that applicative may-similarity is a congruence. 
Most proofs of this use syntactic 
arguments, but here we use a semantic method developed by Pitts~\cite{pitts1997note}
in the context of domain theory, which we adapt to guarded recursion. As a first step we construct a denotational 
semantics of the untyped lambda calculus. 

\subsection{Denotational semantics}
\label{sec:denotational:semantics}

Like the operational semantics, the denotational semantics is parametrised by a monad $\gMon$
equipped with a join-semilattice structure and a delay algebra structure, satisfying (\ref{eq:bind:comm:step}) and 
(\ref{eq:bind:comm:cup}). Closed terms will be interpreted as elements of the type $\D$ defined
by the following equations.
%
\begin{align*}
  \SVal &\defeq \later^\kappa(\SVal \to \gMon(\SVal)) &
  \D &\defeq \gMon(\SVal)
\end{align*}
Here the definition of $\SVal$ should be read as a guarded recursive definition, and states that semantic values
can be considered effectful computations on semantic values, but that this unfolding takes a single computation step.
Define a semantic application $\SApp{}{} : \D \times \D \to \D$ as 
\[
\SApp d{d'} = d \sequencing \lambda f.~ d' \sequencing \lambda v.~ \gstep (\tabs\tickA\kappa{\tapp f\,v})
\]
and using this we define the operational semantics
$\den{-} : \Lambda(n) \to (\SVal)^n \to \D$
where $\Lambda(n)$ is the set of terms with at most $n$ free variables, as
\begin{align*}
 \den{x_i}\rho & = \gnow\,(\rho\,i) &
 \den{\olam{x_{n+1}}M}\rho & = \gnow(\tabs\tickA\kappa{(\lambda d . \den{M}(\rho, d))}) \\
 \den{M \, N}\rho & =\SApp{\den{M}\rho}{(\den{N}\rho)} &
 \den{M \oor N}\rho &= \den{M}\rho \gCup \den{N}\rho
\end{align*}
Note 
that the interpretation of values factor through  $\gnow : \SVal \to \D$ via $\denv{-} : \Val \to \SVal$.

%
%
%
%
%

\begin{theorem}[Soundness] \label{thm:soundness}
  $\gMon(\denv{-})(\eval\,M) = \den{M}$
\end{theorem}

We now specialise $T$ to $\Pmay$ for the following notation and corollary.
If $Q : A \to \Prop$ and $m : \Pmay A$, we shall use the infix notation $m \SconvMay Q$ for $\gLift Qm$,
where $\gLift$ is the lifting of $\Pmay$.

\begin{corollary} \label{cor:soundness}
 The statements $\den M \SconvMay Q$  and 
 $M \BSopmay[\kappa] Q \circ \denv{-}$ are equivalent.
\end{corollary}

\subsection{Relating syntax and semantics}
\label{sec:relating:syntax:semantics}

We now construct a relation between syntax and semantics that will allow us to
use the model to reason about the operational semantics. The relation is similar to
the one constructed by Pitts~\cite{pitts1997note} in the setting of domain theory, but 
whereas Pitts must provide a technical argument for the existence of the relation,
which is far from obvious in the domain theoretic setting, in our setting the relation
exists simply by guarded recursion. 

In this section we specialise the model from the general monad $\gMon$ to $\Pmay$, and 
define two relations, one on values and one on general terms as follows
\begin{align*}
 \SSrel & : \D \times \Lambda \to \Prop \\
 \SSrelV & : \SVal \times \Val \to \Prop \\ 
 d \SSrel M & \,\defeq d \SconvMay \lambda v . (\exists V. M \BSmay V \hmeet v \SSrelV V) \\
 v \SSrelV \olam xM & \,\defeq  \latbind\tickA\kappa{(\forall} v', V' . v' \SSrelV V' 
 \to ({\tapp[\tickA]v}(v')) \SSrel M \subst {V'} x)
\end{align*}
This is well-defined, because unfolding the definition of $\SSrel$ in the definition of $\SSrelV$ 
gives a guarded recursive definition of $\SSrelV$. 
If $\rho : (\SVal)^n$ and $\sigma : \Val^n$ write $\rho \SSrelV \sigma$ to mean 
$\rho_1 \SSrelV \sigma_1 \hmeet \dots \hmeet \rho_n \SSrelV \sigma_n$. 

\begin{lemma}[Fundamental lemma] \label{lem:fundamental}
 If $\rho \SSrelV\sigma$ then $\den M \rho \SSrel M\sigma$.
\end{lemma}

The fundamental lemma is proved by induction on $M$. Using this, one can prove the following 
correspondence between $\SSrel$ and $\appsimmay$

%
%
%
%
%
\begin{lemma} \label{lem:app:sim:is:SSRel}
If $M$ and $N$ are closed terms then $M \appsimmay N$ is equivalent to $\forall\kappa . \den M \SSrel N$. 
\end{lemma}

The left to right direction is proved by showing that $\SSrel$ is upward closed in its second argument. The 
other direction is proved using guarded recursion. Note that as a consequence of 
Lemma~\ref{lem:fundamental} and~\ref{lem:app:sim:is:SSRel} it follows that $\appsimmay$ is a reflexive relation.

%
\begin{theorem} \label{thm:appsim:congruence}
 $\appsimmay$ is a congruence, i.e., if $M \appsimmay N$ and $C[-]$ is a context then also
 $C[M] \appsimmay C[N]$. 
\end{theorem}

\begin{proof}
 Using reflexivity it suffices to show that if $M \appsimmay N$
 and $M' \appsimmay N'$ then $M \, M' \appsimmay N \, N'$, $M \oor M' \appsimmay N \oor N'$
 and $\olam xM \appsimmay \olam x{N}$. The cases of application and choice can be reduced to 
 the statements that if $d \SSrel M$ and $d' \SSrel N$ then $\SApp{d}{d'} \SSrel M \, N$ 
 and $d \cup d'\SSrel M \oor N$, which can be proved by guarded recursion. 
 To prove $\olam xM \appsimmay \olam x{N}$ it suffices to prove that 
 $\later^\kappa (\forall V. M\subst {V} x \appsimmay N\subst {V} x)$. By definition of applicative 
 may-similarity for open terms, however, we know that $M\subst {V}x \appsimmay N\subst {V}x$.
\end{proof}

\section{A powerdomain for must-convergence}
\label{sec:must:pow:dom}

We now introduce our powerdomain construction $\Pmust$ for must-convergence. 
This should have an inclusion $\mustnow : A \to \Pmust(A)$, a join-semilattice 
structure $\mustcup$ and a delay algebra structure $\muststep$. However, 
when considering must-convergence, a term $M \oor N$ diverges if $M$ diverges
even if $N$ converges. To enforce that in our powerdomain we use equations 
to enforce parallel evaluation of subcomputations and stating that terminating
values are postponed until all subcomputations have been evaluated fully:
\begin{align}
 \muststep (x) \mustcup \muststep (y) & = \muststep(\tabs\tickA\kappa{\tapp x \mustcup \tapp y}) \label{eq:must:step:step} \\
 \muststep (x) \mustcup \mustnow (y) & = \muststep(\tabs\tickA\kappa{(\tapp x \mustcup \mustnow (y))}) \label{eq:must:step:now}
\end{align}
These equations (together with the derivable symmetric version of (\ref{eq:must:step:now})) 
allow steps to bubble up the syntax tree, to a normal form consisting of a 
(possibly infinite) sequence of computation steps followed by a finite set of 
values. Following this intuition we define
\[
  \Pmust(A) \defeq \Lift(\Pfin(A))
\]
This has the benefit over, say a HIT given by the equations above, of giving 
direct access to the set of possible values returned by a computation that must 
converge. 

By definition $\Pmust$ carries a delay-algebra structure, and the inclusion of $A$ into
$\Pmust(A)$ can be defined as 
\[
 \mustnow(a) \defeq \liftnow(\{a\})
\]
The join-semilattice can be defined by guarded recursion using the equations 
(\ref{eq:must:step:step}), (\ref{eq:must:step:now}), the symmetrisation of (\ref{eq:must:step:now})
and 
\[
\liftnow x \mustcup \liftnow y \defeq \liftnow (x \cup y)
\]
A natural question is whether $\Pmust$ defines a monad. Since it is the composite 
of two monads, it is sufficient that there is a distributive law of monads, and indeed a natural
candidate is easily defined as 
\begin{align*}
  \dist & : \Pfin \Lift \to \Lift \Pfin & 
 \dist(X) & \defeq\cup_{x \in X} \Lift(\{-\})(x) 
\end{align*}
However, this only defines a distributive law of the monad $\Pfin$ 
over $\Lift$ considered as a functor, not a monad.

\begin{proposition} \label{prop:dist:law}
 Of the four diagrams for distributive laws over monads:
 \[
       \begin{tikzcd}
        \Lift \ar{dr}{\Lift(\{-\})} \ar[swap]{d}{\{-\}} & &
        \Pfin\Pfin\Lift \ar{r}{\Pfin(\dist)} \ar[swap]{d}{\cup} & \Pfin\Lift\Pfin \ar{r}{\dist} & \Lift\Pfin\Pfin \ar{d}{\Lift(\cup)} \\
        \Pfin\Lift \ar{r}{\dist} & \Lift\Pfin \!\!\!\!\!\! & \Pfin\Lift \ar{rr}{\dist} & &  \Lift\Pfin \\
        \Pfin \ar{dr}{\liftnow{}} \ar[swap]{d}{\!\Pfin(\liftnow{}\!)} &   
        & \Pfin\Lift\Lift \ar{r}{\dist} \ar[swap]{d}{\Pfin(\Lmult)} & \Lift\Pfin\Lift \ar{r}{\Lift\dist} & \Lift\Lift\Pfin \ar{d}{\Lmult} \\
        \Pfin\Lift \ar{r}{\dist} & \Lift\Pfin \!\!\!\!\!\!& \Pfin\Lift \ar{rr}{\dist} & &  \Lift\Pfin 
      \end{tikzcd}
 \]
all but the last commute. 
\end{proposition}

A counterexample to the last is 
 $\{\liftstep (\tabs\tickA\kappa{\liftnow(\liftnow x)}), \liftnow (\liftstep \tabs\tickB\kappa{\liftnow x})\}$
 which is mapped by the lower composite to 
$\liftstep (\tabs\tickA\kappa{\liftnow \{x\}}))$
and by the upper to $\liftstep \tabs\tickA\kappa{\liftstep \tabs\tickB\kappa{\liftnow \{x\}}}$.
Note that these only differ by a finite number of computation steps, i.e.,
are equal up to weak bisimilarity. We conjecture 
that this is generally true and that $\Pmust$ is a monad up to weak bisimilarity.

As a consequence, using $\dist$ to define the multiplication of $\Pmust$ does not 
define a monad. Nevertheless, it does define a bind operation. 

\begin{lemma}
 The bind operation induced by $\dist$ maps $f : A \to \Pmust (B)$ 
 and $a : \Pmust (A)$ to 
 \[
   \Lbind aX{\mustcup_{x\in X} f(x)}
 \]
 and satisfies the equations 
  $(\bindnolambda {\mustnow (a)}f)  = f(a)$ and 
  $(\bindnolambda  a{\mustnow}) = a$, and moreover defines a homomorphism of delay-algebras 
  as well as join-semilattices in $a$. It does not satisfy the associativity axiom.
\end{lemma}
Since the associativity axiom is not used in our development, the proofs done in
the previous section for a general monad $\gMon$ carry over to this case, as we shall see. 

\section{Applicative must-simulation}
\label{sec:app:must:bisim}

We now show how our techniques from the may-convergence case apply to 
show that applicative must-similarity is a congruence also in the case of must-convergence.
First we set up the operational semantics. In the case of the standard big-step semantics, 
define the predicate $\BSmust \subseteq \Lambda \times \Pfin (\Val)$ as 
\begin{mathpar}
  \infer{M \BSmust X \\ N \BSmust Y}{M \oor N \BSmust X \cup Y} \and
  \infer{ }{\olam x M \BSmust \{\olam x M\}} \and
  \infer{M \BSmust X \\ N \BSmust Y \\ \forall (\olam y{M'}) \in X, V \in Y.~ M' \subst Vy \BSmust Z_{\olam y{M'},V}}{M\,N \BSmust \cup_{V' \in X, V \in Y} Z_{V',V}} 
\end{mathpar}
The judgement $M \BSmust X$ states that $M$ must converge and that the possible values 
that it can converge to is $X$.

The evaluation function $\eval : \Lambda \to \Pmust(\Val)$
is defined by specialising the general definition given in Section~\ref{sec:app:may:bisim}. 
We also define a relation $M \BSopmust[\kappa] Q$ stating 
that if $M$ terminates, it will terminate to a set of values satisfying $Q : \Pfin(\Val) \to \Prop$. 
Note that $Q$ is a predicate on sets of values, rather than values themselves (as was the case for
$M \BSopmay[\kappa] Q$). This allows us to express properties e.g. by existential quantification over 
outcome values, as needed e.g. in the definition of must-similarity below. 
To define $M \BSopmust[\kappa] Q$, consider first a lifting $\LPred Q$ of predicates $Q : A \to \Prop$ to 
$\Lift A$ defined as 
\begin{align*}
 \LPred Q(\liftnow a) & \defeq Q(a) & \LPred Q(\liftstep a) & \defeq \latbind \tickA\kappa{\LPred Q(\tapp a)}
\end{align*}
and note that this also satisfies items~\ref{item:Lift:comp} and \ref{item:Lift:bind} of Definition~\ref{def:pred:lifting}.
Define 
$M \BSopmust[\kappa] Q \defeq \LPred Q(\eval (M))$. 

The relationship between these two operational
semantics is similar to the one between $\BSmay$ and $\BSopmay[\kappa]$. First define, for 
%
$m : \forall \kappa. \Lift[\kappa] A$ and $a : A$ a termination predicate 
$m \BSmustforall a$ as an inductive family in $\Prop$ like so:
\begin{mathpar}
\infer{m \BSmustforall a}{(\lambda \kappa.\, \liftstep{(\tabs \_ \kappa {m\,\kappa})}) \BSmustforall a}
\and
\infer{ }{(\lambda \_.\, \liftnow a) \BSmustforall a} 
\end{mathpar}

\begin{proposition} \label{lem:eval:BS}
  The statements $M \BSmust V$ and $(\lambda \kappa.\,\eval\,M) \BSmustforall V$ are logically equivalent.
\end{proposition}

\begin{lemma} \label{lem:elim:kappa:Lf}
  Let $A$ be clock irrelevant, $Q^\kappa$ a family over $A$, and $m : \forall \kappa.\, \Lift[\kappa] A$.
  The statements $\forall\kappa . \Lf Q^\kappa {(m\,\kappa)}$ and $m \BSmustforall a \to \forall\kappa . Q^\kappa(a)$ 
  are logically equivalent. As a consequence the statements 
  $M \BSmust V \to \forall\kappa . Q^\kappa(V)$ and 
   $\forall\kappa . M\BSopmust[\kappa] Q^\kappa$ are equivalent.
  \end{lemma}

%

Say that a relation $R$ on closed terms is an applicative must-simulation if $M R N$ implies
\[M \BSmust U \to \exists V.\,N \BSmust V \hmeet \forall (\olam x {N'} \in V).\, \exists (\olam x {M'} \in U).\,
 (\forall (W : \Val) . M'\subst {W} x \,R\, N'\subst {W} x)\]
Define $M \appsimmust[\kappa] N$ by guarded recursion to be
\[
 M \BSopmust[\kappa] \lambda U .\, \exists V. N \BSmust V \hmeet \forall (\olam y{N'} \in V).\,
\exists (\olam x {M'} \in U).\, \laterAbs\tickA\kappa 
 (\forall W . M'\subst {W}x \appsimmust[\kappa] N'\subst {W}x)
\]
This is extended to open terms by defining $M \appsimmust[\kappa] N$ to mean $M\sigma \appsimmust[\kappa] N\sigma$
for all substitutions $\sigma$ mapping all free variables in $M$ and $N$ to closed terms. Write $M \appsimmust N$ for
$\forall\kappa . M \appsimmust[\kappa] N$. 

\begin{lemma} \label{lem:appbisim:great:must}
 $\appsimmust$ is the greatest applicative must-simulation. 
\end{lemma}

Also in the case of the denotational semantics the general case described in 
Section~\ref{sec:denotational:semantics} specialises to $\Pmust$. None of the proofs
or constructions rely on associativity of the bind operation, so also the soundness result 
holds. For our applications of the denotational semantics, however, we need a variant 
of Corollary~\ref{cor:soundness} which applies to predicates on sets of values 
rather than on values themselves. This uses an infix notation
$m \Sconv Q$ for $\LPred Q(m)$.

\begin{corollary} \label{cor:soundness:must}
 The statements $\den M \Sconv Q$  and 
 $M \BSopmust[\kappa] Q \circ \Pfin(\denv{-})$ are equivalent.
\end{corollary}

\subsection{Relating syntax and semantics}
\label{sec:relating:syntax:sem:must}

As in the case of may-convergence we now construct a relation between syntax and semantics. 
To simplify syntax we introduce the lifting of a relation 
$R : X \times Y \to \Prop$ to a relation on 
powersets $\Pfinrel(R) : \Pfin X \times \Pfin Y \to \Prop$ defined as
\[
  \Pfinrel(R)(A, B) = \forall b \in B \exists a\in A . R( a , b)
\]
We define two relations between syntax and semantics by mutual guarded recursion (overwriting 
notation from Section~\ref{sec:relating:syntax:semantics}):
\begin{align*}
 \SSrel & : \D \times \Lambda \to \Prop \\
 \SSrelV & : \SVal \times \Val \to \Prop \\ 
 d \SSrel M & \,\defeq d \Sconv \lambda A . \exists B. M \BS B \hmeet \Pfinrel(\SSrelV)(A,B) \\
 v \SSrelV \olam xM & \,\defeq  \latbind\tickA\kappa{\forall v', V' . v' \SSrelV V' \to ({\tapp v}(v')) \SSrel M \subst {V'} x))}
\end{align*}
If $\rho : (\SVal)^n$ and $\sigma : \Val^n$ write $\rho \SSrel \sigma$ to mean 
$\rho_1 \SSrel \sigma_1 \hmeet \dots \hmeet \rho_n \SSrel \sigma_n$. 

\begin{lemma}[Fundamental lemma] \label{lem:fundamental:must}
 If $\Gamma \vdash M$ and $\rho \SSrel \sigma$ then $\den M \rho \SSrel M\sigma$.
\end{lemma}

The proof of Lemma~\ref{lem:fundamental:must} is by induction on $M$. In particular the
case of application requires some work, and relies on the fact that $\Pfinrel$ respects the monad 
structure of $\Pfin$ in the sense that if $f: X \to \Pfin(X')$ and $g: Y \to \Pfin(Y')$ map 
pairs related in $R : X\times Y \to \Prop$ to pairs related in $\Pfinrel(S)$, then 
the extensions $\overline f: \Pfin (X) \to \Pfin(X')$ and $\overline g: \Pfin (Y) \to \Pfin(Y')$
map pairs related in $\Pfinrel(R)$ to pairs related in $\Pfin(S)$.

\begin{lemma} \label{lem:appsim:SSrel:must}
 $M \appsimmust N$ iff $\forall\kappa . \den M \SSrel N$. 
\end{lemma}

Similarly to the case of may-convergence, this implies that applicative may-similarity is a reflexive relation. From this 
it follows that it is a congruence exactly as in the proof of Theorem~\ref{thm:appsim:congruence}.
%

\begin{theorem}
 $\appsimmust$ is a congruence, i.e., if $M \appsimmust N$ and $C[-]$ is a context then also
 $C[M] \appsimmust C[N]$. 
\end{theorem}

\section{Conclusion}
\label{sec:conclusion}

The constructions of this paper illustrate how the combination of guarded 
recursion with higher inductive types and univalence in Clocked Cubical Type Theory
gives an expressive type theory for reasoning about programming languages. 
In particular, this combination allows arguments known from domain theory involving
constructions such as recursive types to be represented in type theory.
Moreover, the abstract setting of synthetic guarded domain theory allows for these
tools to be used in a much more elementary setting, far from the mathematical
complexity of domain theory. This is particularly clear in the construction of 
the relation $\appsimmay[\kappa]$ which in ordinary domain theory requires a
non-trivial existence argument~\cite{pitts1997note,pitts1996relational}. It also 
appears in our definitions of the guarded powerdomains, which we define much more
directly than the standard constructions in domain theory~\cite{domain:handbook}. 

It is unfortunate that the bind rule for $\Pmust$ is not associative. As mentioned,
this does not affect our constructions, and we conjecture that it is associative 
up to weak bisimilarity, and that this is enough for most purposes. We believe the 
reason for the failure of associativity is that the equality (\ref{eq:must:step:now})
is not algebraic in the sense that it only applies when one side is a value. One way 
to avoid this is to replace (\ref{eq:must:step:step}) and (\ref{eq:must:step:now}) by
an equation of the form 
\[
 \muststep (x) \mustcup y = \muststep(\tabs\tickA\kappa{(\tapp x \mustcup y)}) 
 \]
which means that to evaluate $x \cup y$ takes as many steps as the sum of steps
used to evaluate $x$ and $y$ respectively, rather than the maximum. In particular,
this means that idempotency is lost (but may hold up to weak bisimilarity) and one 
essentially works with finite multisets rather than the standard powerset. 

Future work includes extending to the case of countable non-determinism. This could use
the countable powerset functor, which is also definable as a HIT~\cite{quotientingDelay}.
We believe that the case of may-convergence generalises directly to the countable
case, but in the case of must-convergence the definition of $\mustcup$ as used here
requires deciding if all branches of a computation terminates. We believe this is a 
symptom of a much more fundamental problem, namely that the partiality monad of 
guarded recursion describes termination in finite steps, whereas the must-convergence
predicate for countable non-determinism requires more steps to reach a fixed point. 
Bizjak et al.~\cite{Bizjak:TLCA:2014} observe a similar
problem in the operational setting and solve it using a combination of $\top\top$-lifting
and transfinite induction in the underlying step-indexing model. It would be interesting 
to see if such an approach also applies to type theory.

Finally, it would be interesting to develop a general theory of combinations of algebraic 
effects such as state, exceptions, and non-determinism (as studied here) with guarded 
recursion. The domain theoretic counterparts of these effects are usually described 
algebraically using order-enriched theories~\cite{hyland2006discrete}, but as we have 
seen here, in the setting of guarded recursion the intensional information of the individual
steps allows us to describe the interaction of these effects with recursion in terms of 
ordinary equations. This theory could then give rise to a notion of guarded interaction 
trees~\cite{xia2019interaction}
which would allow also equations between computations across steps as well as 
guarded recursive definitions.

\textbf{Acknowledgements.} We thank the anonymous reviewers for many useful observations 
and suggestions.

\bibliographystyle{eptcs}
\bibliography{paper}



\end{document}

%% file: prelude.tex
\usepackage{amsmath,amsfonts}
\usepackage{amssymb} 
\usepackage{ifxetex}
\usepackage[T1]{fontenc}
\usepackage[utf8]{inputenc}
\usepackage[cal=boondox]{mathalfa}
\usepackage{stmaryrd}
\usepackage{url}
\usepackage{tikz-cd}
\usepackage{amsthm}
\usepackage{mathpartir}

\renewcommand{\L}{\mathsf{L}}

\renewcommand{\vartheta}{\delta}

\newcommand{\Univ}[0]{\mathsf{U}}

\newcommand{\Trunc}[1]{||#1||}

\newcommand{\subst}[2]{[#1/#2]}

\newcommand{\timeless}[1]{{\mathsf{TimeLess}(#1)}}

\newcommand{\force}{\mathsf{force}}

\newcommand{\later}{\triangleright}
\newcommand{\laterAbs}[3]{\latbind{#1}{#2}{#3}}

\newcommand{\fix}{\mathsf{fix}}

\newcommand{\den}[1]{{\llbracket #1 \rrbracket}^\kappa}
\newcommand{\denv}[1]{{\llbracket #1 \rrbracket}_{\mathsf{Val}}^\kappa}

\newcommand\tabs[2]{\lambda (#1\! :\! #2).}

\newcommand\tapp[2][\tickA]{#2\,[#1] }

\newcommand{\tickA}{\alpha}
\newcommand{\tickB}{\beta}
\newcommand\latbind[2]{{\triangleright}\, (#1 \!: \!#2) .}
\newcommand\toksubst[3][\kappa]{\left[#2/#3\right]}

\newcommand{\clocktype}{\mathsf{clock}}

\newcommand{\capp}[2][\kappa]{#2\,[#1]}

\newcommand\hastype[4][]{
#2 \vdash_{#1} #3: #4
}

\newcommand\wfcxt[2][]{#2 \vdash_{#1}}

\newcommand{\defeq}{\mathbin{\overset{\textsf{def}}{=}}}

\newcommand{\I}{\mathbb{I}}

\renewcommand{\phi}{\varphi}

\newcommand{\equi}{\simeq}

\newcommand{\olam}[2]{\olambda #1 . #2}
\newcommand{\olambda}{%
  \mathop{%
    \rlap{$\lambda$}%
    \mkern2mu
    \raisebox{.275ex}{$\lambda$}%
  }%
}
\DeclareMathOperator{\oor}{\mathsf{or}}

\newcommand{\Val}{\mathsf{Val}}

\newcommand{\BS}{\Downarrow_{\Diamond}}

\newcommand{\BSopmay}[1][]{\Downarrow^{#1}_{\Diamond}}
\newcommand{\BSmay}{\Downarrow_{\Diamond}}
\newcommand{\Convmay}{\Downarrow_{\Pmay[]}^\forall}
\newcommand{\BSopmust}[1][]{\Downarrow^{#1}_{\Box}}
\newcommand{\BSmust}{\Downarrow_{\Box}}
\newcommand{\BSmustforall}{\Downarrow_{\Box}^\forall}

\newcommand{\SSrel}[1][\kappa]{\preceq^{#1}}
\newcommand{\SSrelV}[1][\kappa]{\preceq^{#1}_{\mathsf{Val}}}

\newcommand{\appsimmayV}[1][]{\leq^{#1}_{\mathsf{Val}}}

\newcommand{\appsimmust}[1][]{\leq^{#1}_{\Box}}
\newcommand{\appsimmay}[1][]{\leq^{#1}_{\Diamond}}

\newcommand{\dist}{\zeta}
\newcommand{\Prop}{\mathsf{Prop}}

\newcommand{\gMon}{T}
\newcommand{\gLift}{\hat{T}}

\newcommand{\gnow}[1]{\mathsf{pure}\,#1}
\newcommand{\gstep}[1]{\mathsf{step}_{\gMon}\,#1}

\newcommand{\gCup}{\cup}

\newcommand{\maynow}[1]{\mathsf{now}_{\Diamond}\,#1}
\newcommand{\maystep}[1]{\mathsf{step}_{\Diamond}\,#1}

\newcommand{\mustnow}{\mathsf{now}_{\Box}}
\newcommand{\muststep}{\mathsf{step}_{\Box}}
\newcommand{\mustcup}{\cup}

\newcommand{\liftnow}[1]{\mathsf{now}_{\Lift[]}\,#1}
\newcommand{\liftstep}[1]{\mathsf{step}_{\Lift[]}\,#1}

\newcommand{\Lift}[1][\kappa]{\mathsf{L}^{#1}}
\newcommand{\sequencing}{\,\texttt{>{}>=}}

\newcommand{\Lbind}[3]{#1\sequencing_{\Lift[]}\lambda #2 . #3}
\newcommand{\bind}[3]{#1\sequencing\lambda #2 . #3}
\newcommand{\bindnolambda}[2]{#1\sequencing #2}
\newcommand{\Lmult}{\mu_{\Lift[]}}

\newcommand{\LPred}[1][\kappa]{\hat{\Lift[#1]}}
\newcommand{\Lf}{\LPred}

\newcommand{\eval}{\mathsf{eval}}

\newcommand{\D}[1][\kappa]{D^{#1}}
\newcommand{\SVal}[1][\kappa]{\mathsf{SVal}^{#1}}
\newcommand{\SApp}[2]{#1\cdot #2}

\newcommand{\Sconv}[1][\kappa]{\downarrow^{#1}}
\newcommand{\SconvMay}[1][\kappa]{\downarrow^{#1}_{\Diamond}}
\newcommand{\Pfin}{{\mathsf{P}_{\mathsf{f}}}}
\newcommand{\Pfinrel}{{\overline{\mathsf{P}}_{\mathsf{f}}}}

\newcommand{\Pmay}[1][\kappa]{{\mathsf{P}_{\Diamond}^{#1}}}
\newcommand{\PmayPred}[1][\kappa]{\hat{\mathsf{P}}_{\Diamond}^{#1}}
\newcommand{\Pmust}[1][\kappa]{\mathsf{P}_{\Box}^{#1}}

\newcommand{\inl}{\mathsf{inl}}
\newcommand{\inr}{\mathsf{inr}}

\newcommand{\hmeet}{\land}
\newcommand{\hvee}{\lor}

\newcommand{\isProp}{\mathsf{isProp}}
\newcommand{\isSet}{\mathsf{isSet}}

\newtheorem{theorem}{Theorem}
\newtheorem{lemma}{Lemma}
\newtheorem{proposition}{Proposition}

\newtheorem{definition}{Definition}
\newtheorem{corollary}{Corollary}